\documentclass[12pt]{article}
\usepackage[utf8]{inputenc}
\usepackage{amsthm}
\usepackage{amsmath}
\usepackage{amsfonts}
\usepackage{graphicx}
\usepackage{latexsym}
\usepackage{amssymb}
\usepackage{array}
\usepackage{upref}
\usepackage{mathtools}
\usepackage{mathrsfs}

\usepackage{diagbox}

\usepackage[linesnumbered,ruled]{algorithm2e}

\usepackage{scalerel,stackengine}
\newcommand\reallywidehat[1]{%
	\savestack{\tmpbox}{\stretchto{%
			\scaleto{%
				\scalerel*[\widthof{\ensuremath{#1}}]{\kern-.6pt\bigwedge\kern-.6pt}%
				{\rule[-\textheight/2]{1ex}{\textheight}}
			}{\textheight}%
		}{0.5ex}}%
	\stackon[1pt]{#1}{\tmpbox}%
}

 \usepackage[colorlinks=true]{hyperref}
  \hypersetup{urlcolor=blue, citecolor=red}

\DeclareMathAlphabet{\mathbfsl}{OT1}{ppl}{b}{it} 

\newcommand{\edge}[3]{#1\overset{#3}{\longrightarrow}#2}

\makeatletter
 \DeclareRobustCommand{\nsbinom}{\genfrac[]\z@{}}
 \makeatother

\newcommand{\cB}{{\cal B}}
\newcommand{\cC}{{\cal C}}

\newcommand{\cP}{{\cal P}}

\newcommand{\cR}{{\cal R}}
\newcommand{\cS}{{\mathcal{S}}}

\newcommand{\cX}{{\cal X}}
\newcommand{\cY}{{\cal Y}}

\newcommand*{\NN}{\mathbb{N}}

\newcommand{\linadd}{\kern1pt\mbox{\small$\boxplus$}\kern1pt}

\newcommand{\Cref}[1]{Co\-ro\-lla\-ry\,\ref{#1}}

\newtheorem{theorem}{Theorem}[section]

\newtheorem{lemma}[theorem]{Lemma}
\newtheorem{corollary}[theorem]{Corollary}
\newtheorem{example}{Example}

\SetKwInput{KwInput}{Input}                
\SetKwInput{KwOutput}{Output}              

\title{Maximum Length RLL Sequences in de Bruijn Graph}

\date{\today}
\author{\textbf{Yeow Meng Chee$^\text{x}$}, \textbf{Tuvi Etzion$^{\text{x}*}$}, \textbf{Tien Long Nguyen$^\text{x}$},\\
\textbf{Duy Hoang Ta$^\text{x}$}, \textbf{Vinh Duc Tran$^+$}, \textbf{Van Khu Vu$^\text{x}$}\\
{\small $^\text{x}$Dept. of Industrial Systems Engineering and Management, National University of Singapore}\\
{\small $^*$Computer Science Faculty, Technion, Israel Institute of Technology, Haifa 3200003, Israel}\\
{\small $^+$Hanoi University of Science and Technology, Vietnam}\\
{\small {\it ymchee@nus.edu.sg}, {\it etzion@cs.technion.ac.il}, {\it longnt23@nus.edu.sg}},\\
{\small {\it hoang27@nus.edu.sg}, {\it ductv@soict.hust.edu.vn}, {\it isevvk@nus.edu.sg}} 
\thanks{Parts of this work have been presented at the \emph{IEEE International Symposium on Information Theory}, Espoo, Finland, June-July 2022 and the \emph{IEEE International Symposium on Information Theory}, Taipei, Taiwan, June 2023.
The research of T. Etzion was supported in part by the Israeli Science Foundation grant no. 222/19.}
}

\begin{document}

\maketitle

\begin{abstract}
Free-space quantum key distribution requires to synchronize the transmitted and received signals.
A timing and synchronization system for this purpose based on a de Bruijn sequence has been proposed and studied recently for a channel
associated with quantum communication that requires reliable synchronization. To avoid a long period of no-pulse in such
a system on-off pulses are used to simulate a \emph{zero} and on-on pulses are used to simulate a \emph{one}.
However, these sequences have high redundancy and low rate. To reduce the redundancy and increase the rate,
run-length limited sequences in the de Bruijn graph are proposed for the same purpose.
The maximum length of such sequences in the de Bruijn graph is studied and an efficient algorithm to construct a large set of these sequences
is presented. Based on known algorithms and enumeration methods,
maximum length sequence for which the position of each window can be computed efficiently is presented
and an enumeration on the number of such sequences is given.
\end{abstract}

\vspace{0.5cm}


\vspace{0.5cm}


\newpage
\section{Introduction}
\label{sec:intro}

Quantum key distribution are important to prevent quantum computer-based attacks on public key cryptosystems~\cite{ZQLR21}.
In a free-space quantum key distribution, one of the important challenges is
to synchronize the transmitted and received signals accurately.
There are many efforts in designing an efficient and reliable timing and synchronization systems, e.g.~\cite{DCS21,KSSNDBB}.
Unfortunately, the suggested systems suffer from a few disadvantages such as slow transmission if for example a clock is used at either
end of the transmitter and the receiver. To overcome such problems, in~\cite{ZQLR21}
a de Bruijn sequence-based timing and synchronization system is introduced using a beacon with an on-off model. In this model, a sequence
of beacon pulses is used to represent a binary de Bruijn sequence. Hence, once a sub-string of beacon pulses is received, its position is also
determined uniquely. Furthermore, to consider the timing jitter performance, a long period of no-pulses is forbidden. Assume
one pulse slot is used to represent a binary bit. If on-pulse is \emph{one} and off-pulse is \emph{zero}, then
a long run of \emph{zeros} in the sequence, which is a
long period of no-pulses, would impact the timing jitter. In~\cite{ZQLR21}, two pulse slots are suggested
to represent a single bit, on-on (i.e., `11') is a~\emph{one} and on-off (i.e., `10') is a~\emph{zero} so that two consecutive
no-pulses are avoided. However, this scheme requires $2N$ pulse slots to represent a sequence
of length~$N$ in the order $n$ de Bruijn graph and it is required to receive a sub-sequence of $2n$ pulse slots to locate its position,
i.e., if the sequence is long of about $2N=2^{n+1}$ bits (representing a sequence of length $N$ in the de Bruijn graph),
about $2 \log N$ pulse slots are required to locate its position instead of $\log N$ pulse slots
if any sequence of $N$-pulses was permitted (all logarithms in this paper are in base 2).
A~second disadvantage in the proposed scheme is that the sequences of consecutive \emph{zeros} were constrained to length one, while
in reality, a few consecutive no-pulses are permitted subject to a constraint that their length will not be larger than a certain threshold $s$.
Therefore, there is a target to use less redundant pulse slots to achieve both goals, to synchronize accurately, and to avoid long periods of
no-pulses. For this purpose run-length limited (RLL) sequences in the de Bruijn graph are proposed~\cite{CDNTV22,CNTV23}, where one pulse is
represented by one binary bit. An on-pulse is represented by a \emph{one} and an off-pulse is represented by a \emph{zero}.
This scheme is more general and more efficient with lower redundancy and higher rate than the ones in the previous work.
The scheme combines two concepts, RLL sequences and sequences in the de Bruijn graph which will be defined now.
Such sequences are important from engineering point of view. They are formed by using combinatorial
properties of the de Bruijn graph. Moreover, the scheme itself is simple and can be applied easily.

{\bf \emph{RLL sequences}} are binary sequences in which there is an upper bound on the number of consecutive \emph{zeros} in a sequence.
An {\bf $(n,s)$\emph{-word}} is a binary word of length $n$ in which the longest run of consecutive \emph{zeros} is of
length at most $s$. An {\bf $(n,s)$-\emph{sequence}} is a sequence whose windows of length $n$ are distinct $(n,s)$-words.
The family of sequences in which there are no runs of more than $s$ consecutive \emph{zeros}
was extensively studied due to many applications that require such sequences~\cite{Imm90,Imm04}.

The de Bruijn graph of order $n$, $G_n$, was defined first by Nicolaas Govert de Bruijn~\cite{deB46} and in parallel
by Good~\cite{Goo46}. The graph is a directed graph with $2^n$ vertices which are represented by the set
of all binary words of length $n$.
The edges of $G_n$ are represented by the $2^{n+1}$ binary words of length $n+1$.
There is a directed edge $(x_0 x_1 ~ \cdots ~ x_{n-1} x_n)$ from the vertex $(x_0 x_1~\cdots ~ x_{n-1})$ to the vertex
$(x_1 ~ \cdots ~ x_{n-1} x_n)$. A {\bf \emph{span $n$ de Bruijn sequence}} is a cyclic binary sequence in which each binary $n$-tuple
is contained exactly once in a window of length $n$.
The de Bruijn graph and its sequences associated with the graph were studied extensively~\cite{Etz24,Gol17}.

A {\bf \emph{walk}} in the graph is a sequence of directed edges such that
the end-vertex of one edge is the start-vertex of the next edge. The length of a {\bf \emph{walk}} is the number of edges in the walk.
A {\bf \emph{tour}} is a walk in which the first vertex is also the last one.
Any binary sequence can be represented by a walk in $G_n$,
where any $n$ consecutive symbols represent a vertex and any $n+1$ consecutive symbols represent an edge.
A cyclic sequence can be represented by a tour. The consecutive
$n$ symbols and $n+1$ symbols represent a walk with its consecutive vertices and consecutive edges, respectively.
A {\bf \emph{path (cycle)}} in the graph is a walk (tour) with no repeated vertices.
In a cycle, each vertex can be considered as the first vertex of the cycle.
A {\bf \emph{trail} (circuit)} in the graph is a walk (tour) with no repeated edges.
A span~$n$ de Bruijn sequence can be represented by an Eulerian circuit
in $G_{n-1}$, i.e., a circuit which traverses each edge exactly once.
In a circuit, the first vertex is also the last one. It can be also represented by a Hamiltonian cycle in~$G_n$,
i.e., a cycle that visit each vertex of $G_n$ exactly once.
Similarly, any path of length $N$ in $G_n$ can be represented by an acyclic sequence of length $N+n-1$
with no repeated $n$-tuples. Each $n$-tuple is associated with a vertex.
This sequence can be represented also by a trail in $G_{n-1}$, where each $n$-tuple is associated with an edge.
This is demonstrated in Example~\ref{ex:seq_graph}, where also the distinction between cyclic and acyclic sequence
is demonstrated. In the rest of the paper if not mentioned, then the sequence is cyclic, but the results are given for both types
of sequences.

\begin{example}
\label{ex:seq_graph}
For $n=3$, the cyclic sequence $[00011101]$ is a span 3 de Bruijn sequence whose cycle is in $G_3$ is as follows
$$
(000) \rightarrow (001) \rightarrow (011) \rightarrow (111) \rightarrow (110) \rightarrow (101) \rightarrow (010) \rightarrow (100) \rightarrow (000).
$$
In $G_2$ this circuit is
$$
\edge{00}{00}{000} \edge{}{01}{001} \edge{}{11}{011} \edge{}{11}{111} \edge{}{10}{110} \edge{}{01}{101} \edge{}{10}{010} \edge{}{00}{100}
$$
As an acyclic sequence this sequence is written as $0001110100$, i.e., adding two bits (and in general $n-1$ bits) to the cyclic sequence.
For the following path in $G_3$
$$
(000) \rightarrow (001) \rightarrow (011) \rightarrow (111) \rightarrow (110)
$$
its sequence (acyclic) is $0001110$.
\end{example}

One structure that will be used in our exposition and is extensively studied in the literature associated with the
de Bruijn graph is a necklace. A {\bf \emph{necklace of order $n$}} is a cycle in $G_n$ whose length is a divisor of $n$.
If the cycle is of length $n$, then the necklace is of {\bf \emph{full-order}}. If the necklace is of length which is a divisor of $n$
that is smaller than $n$, then the necklace is {\bf \emph{degenerated}}. Each vertex in $G_n$ is represented by a word of length $n$
and in the necklace, we have all cyclic shifts of such a word. The necklaces of order $n$ are the equivalence classes
of the relation defined on words of length $n$, where two words are related if one is a cyclic shift of the other.
If the necklace has length $d$ which is a divisor of~$n$, then it contains $d$ words of length $n$, but
the necklace can be represented by a sequence of length~$d$.
In $G_{n-1}$ such a necklace of order $n$ is also a cycle of length $d$.
The necklace in $G_{n-1}$ is represented by the edges represented by words of length $n$.
The runs of \emph{zeros} in a necklace are considered to be a cyclic run since the necklace is
a cycle. Let {\bf \emph{$(n,s)$-necklace}} be a necklace of order $n$ that does not have a (cyclic) run of more than $s$ consecutive \emph{zeros}.
Such a necklace contains only $(n,s)$-words. There are cyclic and acyclic sequences. An {\bf \emph{acyclic}} sequence of length $k$
is written as $(s_0 s_1 ~ \cdots ~ s_{k-1})$ and if it represents words of length $n$, then it contains $k-n+1$ words and it is
associated with a walk of length $k-n+1$ in~$G_{n-1}$. A~{\bf \emph{cyclic}} sequence of length $k$ is written as
$[s_0 s_1 ~ \cdots ~ s_{k-1}]$ and if it represents words of length~$n$, then it contains $k$ words and it is associated with a tour
of length $k$ in~$G_{n-1}$. A~necklace by its definition is a cyclic sequence.

\begin{example}
Assume $n=6$ and consider the $(6,2)$-necklace with the vertex $(001011)$. This necklace is of full-order and it contains
the six words of length 6 from the cycle
$$
(001011) \rightarrow (010110) \rightarrow (101100) \rightarrow (011001) \rightarrow (110010) \rightarrow  (100101) \rightarrow (001011)
$$
in $G_6$. These six words are cyclic shifts of each other and their necklace is represented for example by $[001011]$.

In $G_5$ this necklace is represented by the cycle
$$
(00101) \rightarrow (01011) \rightarrow (10110) \rightarrow (01100) \rightarrow (11001) \rightarrow (10010) \rightarrow (00101)
$$
where the representation of the edges in $G_5$ is the same as the representation of
the vertices in $G_6$ for these two cycles in $G_5$ and $G_6$, respectively.

The degenerated necklace $[011]$ of order 6 contains the three vertices in $G_6$ which form the cycle of length 3 whose vertices
and edges are given by
$$
(011011) \rightarrow (110110) \rightarrow (101101) \rightarrow (011011) ~.
$$

In $G_5$ this cycle is represented by
$$
(01101) \rightarrow (11011) \rightarrow (10110) \rightarrow (01101)~.
$$

\hfill\quad $\blacksquare $
\end{example}

Now, let

\begin{itemize}
\item $h_{n,s}$ be the number of $(n,s)$-words,

\item $\ell_{n,s}$ be the number of words in all the $(n,s)$-necklaces,

\item $m_{n,s}$ be the maximum length of a cyclic $(n,s)$-sequence.
\end{itemize}

One of the main results of this contribution is to demonstrate that $m_{n,s} \leq \ell_{n,s}$. The existence of cyclic $(n,s)$-sequences
with length $\ell_{n,s}$ are known to exist from the literature with simple constructions~\cite{GaSa18,SWW16},
and hence we have that $m_{n,s} = \ell_{n,s}$.

In this work, a combination of sequences in the de Bruijn graph that have a run-length constraint on the number
of consecutive \emph{zeros}, is discussed. Another combination of sequences in the de Bruijn graph with local run-length
constraint was considered in~\cite{CEKMVVY22}. An upper bound on the maximum length of cyclic and acyclic $(n,s)$-sequences
is proved and sequences which attain this bound are constructed.
By their definition, if $s$ is small, these sequences can be used in free-space quantum key distribution with flexibility in the parameters.
This is a combinatorial concept motivated by an engineering application problem
and can be used in a very simple way for this application.

Let $G_n(s)$ be the subgraph of $G_n$ induced by all the vertices of $G_n$ whose representations are $(n,s)$-words.
Our goal is to find the maximum length path and maximum length cycle in $G_n(s)$.
The definition of $G_n(s)$ implies that the edges of this graph also do not have a run of \emph{zeros}
whose length is greater than $s$. Therefore, all the $(n,s)$-words are represented by the edges of $G_{n-1}(s)$.
This property will be important in $G_{n-1}(s)$. This implies that a maximum length $(n,s)$-sequence is a trail of
maximum length in $G_{n-1}(s)$ and also a path of maximum length in $G_n(s)$.
For the upper bound on the maximum length of such a sequence,
a trail in $G_{n-1}(s)$ (which is a circuit if the sequence is cyclic) will be considered. For the lower bound on the maximum
length of such a sequence and constructing many such sequences, paths in $G_n(s)$ (which is a cycle if the sequence is cyclic).
This is unique as usually it is not required to use different orders of the graph for similar tasks.

\begin{example}

The graphs $G_3$, $G_3(1)$ and $G_3(2)$ are depicted in Fig.~\ref{fig:DB_G3}.
It is readily verified that a maximum length circuit in $G_3(1)$ has length 7 and it is associated
with a $(4,1)$-sequence. A maximum length circuit in $G_3(2)$ has length 12 and it is associated
with a $(4,2)$-sequence. A maximum length cycle in $G_3(1)$ has length 4 and it is associated with
a $(3,1)$-sequence. A maximum length cycle in $G_3(2)$ has length 7 and it is associated with
a $(3,2)$-sequence.

\begin{figure}[ht]
\vspace{1.2cm}
\begin{picture}(80,80)(-95,60)
\includegraphics[width=10cm]{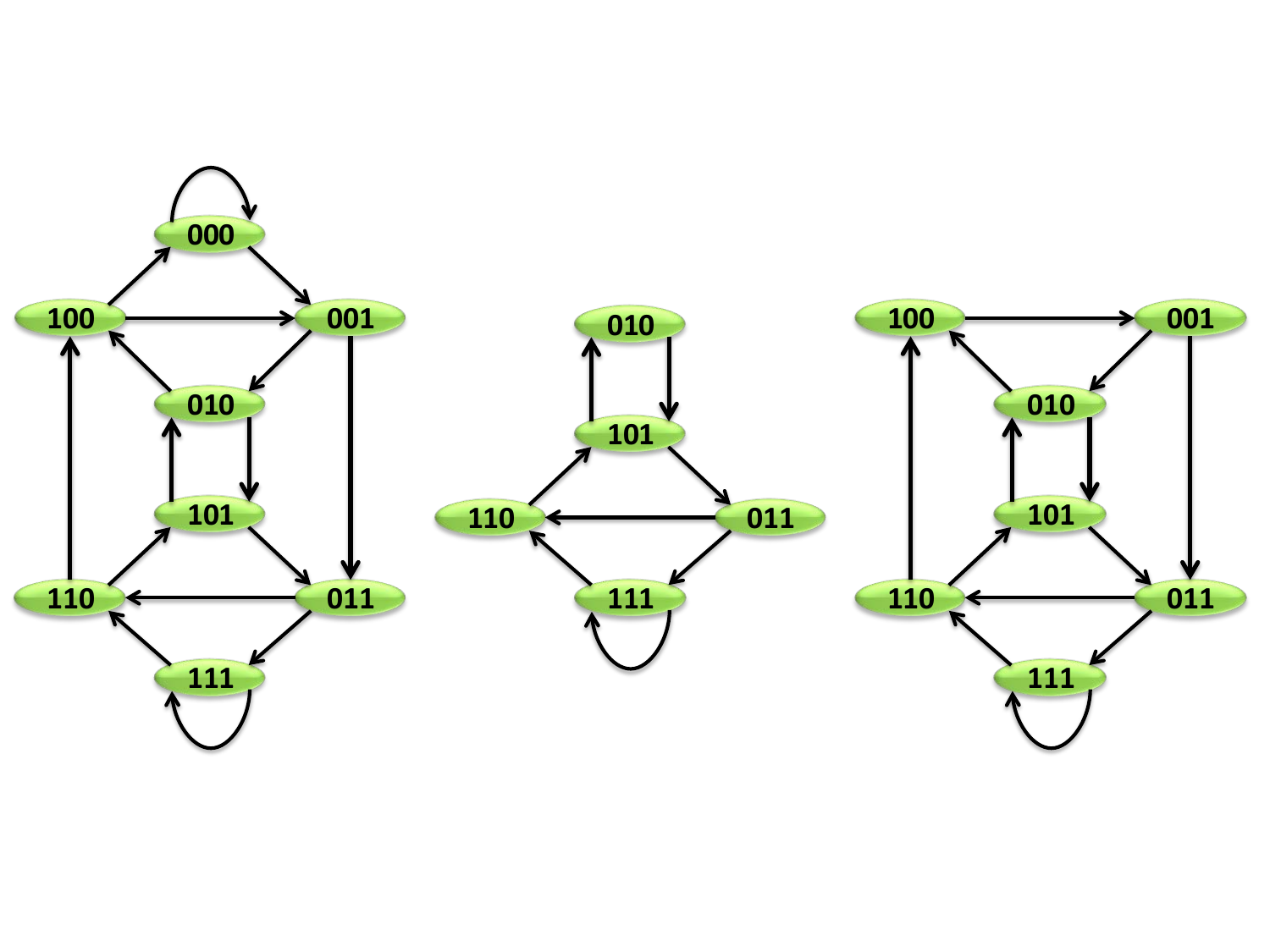}
\end{picture}
\vspace{0.6cm}
\caption{The de Bruijn graph $G_3$ on the left, $G_3(1)$ in the middle and $G_3(2)$ on the right.}
\label{fig:DB_G3}
\end{figure}

\hfill\quad $\blacksquare $
\end{example}

Note, that a cyclic
sequence of length $N$ with no repeated $n$-tuples has $N$ binary digits, $N$~vertices, and $N$ edges in $G_n$ and also in $G_n(s)$.
A related acyclic sequence with no repeated $n$-tuples has $N+n-1$ binary digits, $N$ vertices, and $N-1$ edges in $G_n$ and $G_n(s)$.
Henceforth, unless stated otherwise, all the $(n,s)$-sequences that will be considered are cyclic.

\begin{example}
\label{ex:not_by_rule}
The following cycle in $G_5 (2)$
$$
(00110) \rightarrow (01101) \rightarrow (11010) \rightarrow (10101) \rightarrow (01011) \rightarrow  (10111) \rightarrow (01111) \rightarrow (11111)
$$
$$
\rightarrow (11110) \rightarrow (11101) \rightarrow (11011) \rightarrow (10110) \rightarrow  (01100) \rightarrow (11001) \rightarrow (10010) \rightarrow (00101)
$$
$$
\rightarrow (01010) \rightarrow (10100) \rightarrow (01001) \rightarrow (10011) \rightarrow (00110)
$$
forms the $(5,2)$-sequence
$$
[00110101111101100101]
$$
of length 20. The sequence contains two words $(00110)$ and $(01100)$ which are two consecutive $(5,2)$-words not in a $(5,2)$-necklace
since their necklace contains the words $(00011)$ and $(11000)$, that are not $(5,2)$-words.

\hfill\quad $\blacksquare $
\end{example}

The maximum length $(5,2)$-sequence has length 21. Hence,
Example~\ref{ex:not_by_rule} indicates that there could be long (almost the same length as the maximum length or maybe of maximum length)
$(n,s)$-sequences which contain $(n,s)$-words that are not contained in
$(n,s)$-necklaces. The possibility of such maximum length $(n,s)$-sequences with $(n,s)$-words which are not contained
in $(n,s)$-necklaces cannot be ruled out by simple verification.
This implies that it is not straightforward to know the maximum length of an $(n,s)$-sequence.
Moreover, the length of a maximum length acyclic $(n,s)$-sequence is not obtained from the maximum length
$(n,s)$-sequence by adding $n-1$ bits (as done for a de Bruijn sequence), since more $(n,s)$-words can be added to the cyclic sequence.

Observing that RLL sequences can be applied efficiently for quantum key distribution is the first contribution of this paper from an
engineering point of view. This is the motivation for this paper.
The main contributions of this work can be summarized as follows:
\begin{enumerate}
\item A proof for an upper bound on the length of the constructed cyclic and acyclic sequences, i.e.,
it is proved that $m_{n,s} \leq \ell_{n,s}$. It was known before that there are
$(n,s)$-sequences of length $\ell_{n,s}$, but it was not known that these are the $(n,s)$-sequences of maximum length.
Our proof implies that $m_{n,s} = \ell_{n,s}$.
The proof also implies that a maximum length cyclic $(n,s)$-sequence
contains exactly all the words of all the $(n,s)$-necklaces. Related results are obtained for acyclic $(n,s)$-sequences.

\item An efficient construction for a large set of such sequences using a tailored designed storage that satisfies all the requirements of
$(n,s)$-sequences. Constructions for one $(n,s)$-sequence was known before, but these constructions cannot be adapted for
constructing a large set of $(n,s)$-sequences.
\end{enumerate}

The contributions of the paper are presented in Sections~\ref{sec:upper_max} and~\ref{sec:construct_max}.
Sections~\ref{sec:num_seq} and~\ref{sec:decode} are results which are either straightforward consequences
or were obtained in previous publications. These results are presented for completeness.
For simplicity we usually assume that $s < n-1$ since if $s \geq n$ then an span $n$ de Bruijn sequence is an $(n,s)$-sequence
and if $s=n-1$, then a shortened span $n$ de Bruijn sequence (one \emph{zero} is removed from the run of $n$ consecutive \emph{zeros})
is an $(n,s)$-sequence.

The rest of the paper is organized as follows. In Section~\ref{sec:upper_max} we consider an upper bound on the maximum length
of $(n,s)$-sequences which are cycles in $G_n$ (and also in $G_n(s)$).
We prove that this bound is equal to the number of words in the $(n,s)$-necklaces, i.e., $m_{n,s} \leq \ell_{n,s}$
and it can be attained only by the vertices of these necklaces.
The trail in $G_{n-1}(s)$ with the maximum length can have another $s$ edges (and $s$ vertices) from $G_{n-1}(s)$.
The detailed proof is based on the necklaces obtained by the edges of $G_{n-1} (s)$.
In Section~\ref{sec:construct_max} it will be proved that the upper bound which was derived in Section~\ref{sec:upper_max} can be attained
by many $(n,s)$-sequences which can be constructed efficiently.
This is proved by considering Hamiltonian cycles in a subgraph of $G_n (s)$ that contains only the vertices of the
$(n,s)$-necklaces. By adding $s$ vertices (and $s$ edges) of $G_{n-1}(s)$
from necklaces which contain a cyclic run with $s+1$ \emph{zeros} a longer acyclic $(n,s)$-sequence is obtained.
An efficient algorithm to construct these sequences will be presented. Although the type of algorithm which is presented is
not new, the choice of keys to form a very large set of such sequences is new.
In Section~\ref{sec:num_seq} enumeration of the number of $(n,s)$-necklaces is done
and a formula for the length of a maximum length $(n,s)$-sequence is derived.
In Section~\ref{sec:decode} a maximum length $(n,s)$-sequence from which the position of each $n$-tuple can be efficiently decoded
is constructed. The construction is an adaptation of a well-known method.
The section considers the related literature on these types of constructions.
Conclusion, several possible generalizations, and future research are discussed in Section~\ref{sec:conclude}.

\section{The Maximum Length of an $(n,s)$-Sequence}
\label{sec:upper_max}

In this section, we will present an upper bound on the length of $(n,s)$-sequences.
To make the non-trivial proof simpler, it is broken to a sequence
of claims which lead step-by-step to the main result.
Let $\cC$ be a circuit, of maximum length in~$G_{n-1}$, which does not contain an edge with $s+1$ consecutive \emph{zeros} in its representation.
This circuit is a cycle in $G_n$ with no vertex having $s+1$ consecutive \emph{zeros} in its representation and hence $\cC$
represents an $(n,s)$-sequence.
To find an upper bound on the length of $\cC$ we will try to remove the edges from $G_{n-1}(s)$ which are not contained in $\cC$.
These edges are associated with vertices whose in-degree is not equal to their out-degree.
The following lemmas are immediate consequences from the binary representation of the vertices, their in-edges, and out-edges.

\begin{lemma}
\label{lem:in-degree=1_s}
The vertices in $G_{n-1} (s)$ with in-degree one are those with the prefix $0^s1$.
\end{lemma}

\begin{lemma}
\label{lem:out-degree=1_s}
The vertices in $G_{n-1} (s)$ with out-degree one are those with the suffix $10^s$.
\end{lemma}

The following lemmas are simple observations from the definitions.

\begin{lemma}
All the words in a necklace of order $n$ with at least two disjoint runs of $s+1$ or more \emph{zeros} have at least one acyclic run
with $s+1$ consecutive \emph{zeros}.
\end{lemma}

\begin{corollary}
\label{cor:s+1_run}
All the words of each necklace of order $n$ with at least two disjoint runs of $s+1$ \emph{zeros} are not contained in the circuit $\cC$.
In particular all the words of a degenerated necklace with at least one run of $s+1$ or more \emph{zeros} are not contained in $\cC$.
\end{corollary}

\begin{lemma}
\label{lem:more_2s+1}
If a necklace of order $n$ has a word with a run of at least $2s+1$ \emph{zeros}, then each word on this necklace has at least one
run with more than $s$ \emph{zeros}.
\end{lemma}

\begin{corollary}
All the words of each necklace of order $n$ with a run of at least $2s+1$ \emph{zeros} are not contained in the circuit $\cC$.
\end{corollary}


Henceforth, a string which starts and ends with a \emph{one} and has no run of more than $s$~\emph{zeros}
will be called an {\bf \emph{$s$-ones string}}.

\begin{lemma}
\label{lem:s+kless}
In $G_{n-1}$ a necklace of order $n$ that contains a unique run with more than $s$~\emph{zeros} and the length of this
run is $s+k$, $1 \leq k \leq s$, contains $n-s+k-1$ edges with a run of more than $s$ \emph{zeros}.
\end{lemma}
\begin{proof}
Consider a necklace of order $n$ that contains a unique run with $s+k$ \emph{zeros}, $1 \leq k \leq s$.
The edges in the necklace that do not have a run with more than $s$ \emph{zeros}
are of the form $0^i X0^{s+k-i}$, $k \leq i \leq s$, where $X$ is an $s$-ones string of length $n -s - k$.
There are $s-k+1$ such edges and a total of $n$ edges in the necklace
and hence the total number of edges with a run of at least $s+1$ \emph{zeros} is $n-(s-k+1)=n-s+k-1$.
\end{proof}

\begin{lemma}
\label{lem:pathLs-i}
Any edge in $G_{n-1}(s)$ of the form $e=0^s X 0^k$, where $1 \leq k \leq s$, and $X$ is an $s$-ones string of length $n-s-k$, yields
a path, which starts at the vertex $v=(0^s X 0^{k-1})$, whose length is at least $s-k+1$, its edges do not have a run of more
than $s$ \emph{zeros}, and these edges are not contained in $\cC$.
\end{lemma}
\begin{proof}
The in-vertex of the edge $e=(v,u)$ has the form $v=(0^s X 0^{k-1})$. By Lemma~\ref{lem:in-degree=1_s},
the in-degree of the vertex $v$ is one and hence any circuit of $G_{n-1}(s)$ can contain at most one of its two out-edges.
Remove one of its out-edges $v \rightarrow v_1$ that is not contained in $\cC$ from $G_{n-1}(s)$. Now, $v_1$ has in-degree one and hence $\cC$
contains at most one of its out-edges. Remove the edge which is not contained in $\cC$ from $G_{n-1}(s)$.
The process continues until the out-degree
of the vertex that is reached is already one. The smallest number of edges removed by this process is $s-k+1$ since the shortest
path from~$v$ to such a vertex $(0^{k-1} X 0^s )$ is of length $s-k+1$ (by Lemma~\ref{lem:out-degree=1_s} the vertex $(0^{k-1} X 0^s )$
has out-degree one.).
\end{proof}

For a vertex of the form $v=(0^s X 0^{k-1})$, where $1 \leq k \leq s$ and $X$ is a $s$-ones string,
in $G_{n-1}(s)$, let $\cP(v)$ denote the sub-path of length $s-k+1$ of the path
in $G_{n-1}(s)$, that starts in the vertex~$v$ and ends at a vertex $v'$ whose out-degree one (as described in the
proof of Lemma~\ref{lem:pathLs-i}) and whose edges are not contained in $\cC$.
The sub-path $\cP(v)$ will be called the {\bf \emph{deleted path}} of~$v$.

\vspace{0.2cm}

\noindent
{\bf Remark:} Note that in this section we are considering trails and circuits in $G_{n-1}(s)$, but some of these trails, e.g.,
in the proof of Lemma~\ref{lem:pathLs-i} are paths and they are mentioned as such.

\noindent
{\bf Remark} If $v=(0^s X 0^{k-1})$, then its deleted path is of length $s-k+1$. We will prove later that this
path contains exactly all the vertices that start with $v$ which were deleted as described in the proof of Lemma~\ref{lem:pathLs-i}.

\begin{lemma}
\label{lem:disjoint_paths}
In $G_{n-1}(s)$, the $s-i+1$ vertices in a deleted path $\cP(v)$ of $v = (0^s X_1 0^i)$ are disjoint from the
$s-j+1$ vertices in a deleted path $\cP(u)$ of $u=(0^s  X_2  0^j)$, where $X_1$ and $X_2$ are two distinct $s$-ones strings
and $i,j \geq 0$.
\end{lemma}
\begin{proof}
The first $s-i$ edges in a deleted path that starts in the vertex $(0^s X_1 0^i)$ are exactly the edges either in the path
$$
0^s X_1 0^s
$$
or in a path of the form
$$
0^s X_1 0^{i+k} 1 Y,
$$
where $0 \leq k < s-i$, and the length of $Y$ is $s-i-1-k$.

Similarly, the first $s-j$ edges in a deleted path that starts in the vertex $(0^s X_2 0^j)$ are exactly the edges either in the path
$$
0^s X_2 0^s
$$
or in a path of the form
$$
0^s X_2 0^{j+k} 1 Z,
$$
where $0 \leq k < s-j$, and the length of $Z$ is $s-j-1-k$.

The binary representation of any vertex from the first $s$ vertices in these four paths starts in one of the \emph{zeros} of the first
run of $s$ \emph{zeros}. Hence, a common vertex for these two paths implies that $X_1 = X_2$, a contradiction.
Thus, the claim of the lemma follows.
\end{proof}

\begin{corollary}
\label{cor:unique_del}
There exists unique deleted path $\cP(v)$ of length $s-k+1$ from the vertex $v=(0^s X 0^{k-1})$.
All the deleted paths have distinct vertices.
\end{corollary}
\begin{proof}
By the proof of Lemma~\ref{lem:pathLs-i} and the definition of the deleted path there exists a unique deleted path $\cP(v)$.
By Lemma~\ref{lem:disjoint_paths} all these deleted paths have distinct vertices.
\end{proof}

The sequence of lemmas that were proved lead to the main results of this section.

\begin{theorem}
\label{thm:upper_bound_length}
For any $1 \leq s < n$ we have that $m_{n,s} \leq \ell_{n,s}$.
\end{theorem}
\begin{proof}
Let $\cB$ be a necklace of order $n$ with a run of more than $s$ consecutive \emph{zeros}.
By Lemma~\ref{lem:more_2s+1} all the words in a necklace
with a run of at least $2s+1$ \emph{zeros} have at least one run with more than $s$ \emph{zeros} and hence
all the vertices of these necklaces are not contained in~$\cC$. Similarly by Corollary~\ref{cor:s+1_run} all
the words of a degenerated necklace with at least one run of more than $s$ \emph{zeros}
have at least one run with more than $s$ \emph{zeros} and hence
all the vertices of these necklaces are not contained in $\cC$.

If the longest run of consecutive \emph{zeros} in $\cB$ is $s+k$, $1 \leq k \leq s$, then by Lemma~\ref{lem:s+kless} there are
$n-s+k-1$ edges in the necklace with a run of more than $s$ consecutive \emph{zeros}. Such a necklace has an edge $e$ in $G_{n-1}(s)$
of the form $e=0^s X 0^k$, where $1 \leq k \leq s$ and $X$ is an $s$-ones string,
whose first vertex is $v=(0^s X 0^{k-1})$ and by Lemma~\ref{lem:pathLs-i}
it yields a deleted path $\cP(v)$ whose length is $s-k+1$. Together, $\cB$ ($n-s+k-1$ forbidden edges) and the associated
deleted path $\cP(v)$ ($s-k+1$ forbidden edges) and possibly more edges that were deleted, we
have at least $(n-s+k-1)+(s-k+1)=n$ edges which are not contained in $\cC$.
By Lemma~\ref{lem:disjoint_paths} all these deleted paths have distinct vertices
and hence the number of edges of
$G_{n-1} (s)$ which are not contained in $\cC$ is at least the number of edges in the necklaces of order $n$ which contain
edges with a run of more than $s$ consecutive \emph{zeros}.

Thus, the length of the circuit $\cC$ in $G_{n-1}(s)$ is at most the number of edges in all the $(n,s)$-necklaces in $G_{n-1}(s)$,
i.e., $m_{n,s} \leq \ell_{n,s}$.
\end{proof}

As will be described in Sections~\ref{sec:construct_max} and~\ref{sec:decode}, it is well known that there exists $(n,s)$-sequence
which contains all the words in the $(n,s)$-necklaces~\cite{GaSa18,SWW16} and hence we have the following consequence.

\begin{corollary}
For any $1 \leq s < n$ we have that $m_{n,s} = \ell_{n,s}$.
\end{corollary}

\begin{corollary}
\label{cor:del_path_one}
The deleted path $\cP(v)$ of the vertex $v=(0^s X 0^{k-1})$ is of length $s-k+1$. This path is the unique path
from the vertex $v=(0^s X 0^{k-1})$ to the vertex $(0^{k-1} X 0^s)$,
where $1 \leq k \leq s$ and $X$ is a $s$-ones string.
\end{corollary}

\begin{corollary}
\label{cor:exact_neck}
A circuit of length $\ell_{n,s}$ contains exactly all the words of the $(n,s)$-necklaces.
\end{corollary}
\begin{proof}
For the analysis of Theorem~\ref{thm:upper_bound_length} we have to add the consequences of Corollaries~\ref{cor:unique_del}
and~\ref{cor:del_path_one} (the uniqueness of the deleted paths) and
observe that a circuit of length $\ell_{n,s}$ cannot contain any $(n,s)$-word which is not part of an $(n,s)$-necklace.
\end{proof}


What about the length of the maximum length trail $\cP$ in $G_{n-1}(s)$? All the arguments we have used so far hold
also for a trail with one exception. The trail can contain two vertices $v$ and $u$, where the in-degree of $v$ is one
and its out-degree is two; the in-degree of $u$ is two and its out-degree is one.
In such a scenario, the first vertex in the trail is $v$ and the last vertex in the trail is~$u$. In this trail
we must have the two vertices $u$ and $v$ also as internal vertices (the trail starts at $v$ will arrive at $u$
continue to $v$ and ends at $u$). We can apply the process in the proof of Lemma~\ref{lem:pathLs-i} on all the vertices
in $G_{n-1}(s)$ except for $v$. After all the deleted paths are removed from $G_{n-1}(s)$ if we continue and apply
the process in the proof of Lemma~\ref{lem:pathLs-i} on $v$ we will obtain a deleted path $\cP(v)$ which starts with~$v$ and
ends at $u$. This deleted path can be added to the maximum length cyclic $(n,s)$-sequence to obtain a maximum length acyclic
$(n,s)$-sequences. It can be either at the beginning of the path or at the end of the path.
In other words, for example we have two vertices $v_1$ and $v_2$ for which we have the edges
$v \rightarrow v_1$ and $v \rightarrow v_2$ in $G_{n-1} (s)$ and the in-degree of~$v$ is one.
Such vertices have the form $v=(0^s X)$, $v_1 = (0^{s-1} X0)$, and $v_2 = (0^{s-1} X 1)$.
Note, that $v$~(a word of length $n-1$) is a vertex contained in the $(n,s)$-necklace $[1 0^s X]$
and $v$ is also a vertex contained in the necklace $[0^{s+1} X]$ which is not an $(n,s)$-necklace,
but $s$~words of the necklace are $(n,s)$-words.
The maximum length circuit in $G_{n-1}(s)$ cannot contain both edges $ v \rightarrow v_1$ and $v \rightarrow v_2$
since the in-degree of $v$ is one.
But, a trail can contain both edges if one of them will be the first edge in the trail.
In other words, if both edges are
on the trail~$\cP$, then since they have the same unique predecessor in $G_{n-1}(s)$, it follows that one of them
must be the first edge in $\cP$. It can be only at the beginning of a sub-path of $\cP$ which was excluded
by Lemma~\ref{lem:pathLs-i}. The length of this sub-path in Lemma~\ref{lem:pathLs-i} is $s$ and its structure is as follows:
$$
(0^s X) \rightarrow (0^{s-1} X0) \rightarrow ~ \cdots ~ \rightarrow (0 X 0^{s-1}) \rightarrow (X 0^s)~,
$$
where all the $s$ edges are $(n,s)$-words that are contained in the necklace $[0^{s+1} X]$ which is not an $(n,s)$-necklace.
Thus, we have the following conclusion.

\begin{theorem}
\label{thm:upper_bound_path}
If $s < n-1$, then the length of a trail $\cP$ of maximum length in $G_{n-1}(s)$ is at most $\ell_{n,s} + s$
(if $s \geq n-1$, then its length is is at most $\ell_{n,s}$).
The length of the associated acyclic $(n,s)$-sequence is $\ell_{n,s} + s +n-1$.
In any such sequence we have all the words in the $(n,s)$-necklaces and one deleted path of length $s$.
\end{theorem}

\begin{example}
If $n=3$ and $s=2$, then a $(3,2)$-sequence of maximum length is
$$
[0010111]
$$
and an acyclic $(3,2)$-sequence of maximum length is
$$
001011100~.
$$
Both sequences have seven $(3,2)$-words and they contain all the $(3,2)$-necklaces.

If $n=5$ and $s=2$, then an $(5,2)$-sequence of maximum length has length 21 and one such sequence is
$$
[110010100111010110111]~.
$$
An acyclic $(5,2)$-sequence has length 27 (4 additional bits, $1100$ which are also the first 4 bits for the acyclic representation,
and 2 additional bits from 2 words, $(00110)$ and $(01100)$,
which are not contained in $(5,2)$-necklaces) as follows
$$
001100101001110101101111100~.
$$

\hfill\quad $\blacksquare $
\end{example}

A maximum length acyclic $(n,s)$-sequence (trail) in $G_{n-1}(s)$ contains $s$ words which are not contained in
a maximum length $(n,s)$-sequence (circuit).
It was explained how to add them to a maximum length $(n,s)$-sequence. The maximum length trail can be constructed from
a maximum $(n,s)$-sequence. One deleted path $\cP(v)$ of length $s$ is taken from a necklace and it is added to obtain a maximum length
acyclic $(n,s)$-sequence. It should be noted that the maximum length $(n,s)$-sequence either starts or ends with a
deleted path. If words of a deleted path are added in other places a large cycle can be obtain, but it will fall
short of~$\ell_{n,s}$. This is demonstrated in the following example.

\begin{example}
The sequence
$$
[00110101111101100101]
$$
of Example~\ref{ex:not_by_rule} is a $(5,2)$-sequence of length 20 which contains two $(5,2)$-words $(00110)$ and $(01100)$ which are
not contained in $(5,2)$-necklaces. This sequence falls short by one from the upper bound of a maximum length $(5,2)$-sequence.
The deleted path which was removed from $G_4(2)$ is
$$
(0011) \rightarrow (0111) \rightarrow (1110) \rightarrow (1100)
$$
and its words of length 5 are contained in the $(5,2)$-necklace $[00111]$.

The sequence
$$
001101011111011001010011100
$$
of length 27 is an acyclic $(5,2)$-sequence which contains 23 $(5,2)$-words, two of which are not contained in an
$(5,2)$-necklace. This sequence attains the upper bound of a maximum length acyclic $(5,2)$-sequence presented
in Theorem~\ref{thm:upper_bound_path}. In this sequence
the deleted path was added at the end of the sequence and hence all the words of the $(5,2)$-necklaces are contained in this trail.

\hfill\quad $\blacksquare $
\end{example}

\section{Construction of Sequences of Maximum Length}
\label{sec:construct_max}

In this section, we will concentrate on constructing $(n,s)$-sequences of maximum length.
We will show that there exist many $(n,s)$-sequences which attain the upper bound of Theorem~\ref{thm:upper_bound_length}
and hence these sequences are maximum length $(n,s)$-sequences.
In the literature there are many efficient algorithm to generate one de Bruijn sequence, but these algorithms cannot
be used to generate many de Bruijn sequences with the same efficiency.

We will design a method and an efficient algorithm to generate a large class of maximum length $(n,s)$-sequences. The algorithms
which will be designed are implemented for constructing the next bit (successor rule) given the last $n$ constructed bits of the sequence.
Such algorithms for generating de Bruijn sequences are well documented in the literature starting with the work of Fredricksen~\cite{Fre70,Fre72,Fre75} who was the first to consider efficient generation of a large set of these cycles.
In~\cite{EtLe84} a very large set of such cycles were generated. Another efficient algorithm to generate a large set of de Bruijn
sequences based on a set of primitive polynomials whose degrees are co-prime was presented by Li, Zeng, Li, Helleseth, and Li~\cite{LZLHL16}.
Many other algorithms with this flavor were
designed later and general frameworks for algorithms to generate one sequence by different
successor rules were given in~\cite{GSWW18} and~\cite{GSWW20}
which extended the framework to other similar structures known as universal cycles.
But, these frameworks are not designed for a construction of a large set of sequences.
As there are a few strategies to generate many sequences we concentrate on one in which the number of generated
sequences is $2^K$ when $K$ bits are stored. Each different assignment to these $K$ bits yields a different maximum length $(n,s)$-sequence.
The amount of required storage is determined by the user subject to the value of $n$.
We start by describing the general method and continue with an algorithm to construct a very large class of maximum length $(n,s)$-sequences.
The method and its algorithm is a modification and a generalization of the one
to generate de Bruijn sequences presented in~\cite{EtLe84}.
While the upper bound on the length of a maximum length $(n,s)$-sequences was based on an analysis of trails in $G_{n-1}(s)$,
the $(n,s)$-sequences which will be generated by this method are cycles in $G_n(s)$.
The basic principle in the method is taking all the $(n,s)$-necklaces and merging them into one cycle.
For a vertex $X=(x_1 x_2 ~ \cdots ~ x_{n-1} x_n)$ in $G_n$, its {\bf \emph{companion}} $X'$
is defined by $X'=(x_1 x_2 ~ \cdots ~ x_{n-1} \bar{x}_n)$, where $\bar{x}$ is the binary complement of $x$.
Two vertex-disjoint cycles in $G_n$, $\cC_1$ and $\cC_2$, are joined together into one cycle, if they contain a pair
of companion vertices $X$ on $\cC_1$ and $X'$ on $\cC_2$. If $Y \rightarrow X$ is an edge on $\cC_1$
and $Z \rightarrow X'$ is an edge on $\cC_2$, then keeping all the edges of $\cC_1$ and $\cC_2$ except for these two edges
and adding the edges $Y \rightarrow X'$ and $Z \rightarrow X$, will merge the two cycles into one cycle.
This is depicted in Fig.~\ref{fig:cross-join}.

\begin{figure}[ht]
\vspace{0.0cm}
\begin{picture}(80,80)(-110,100)
\includegraphics[width=9.0cm]{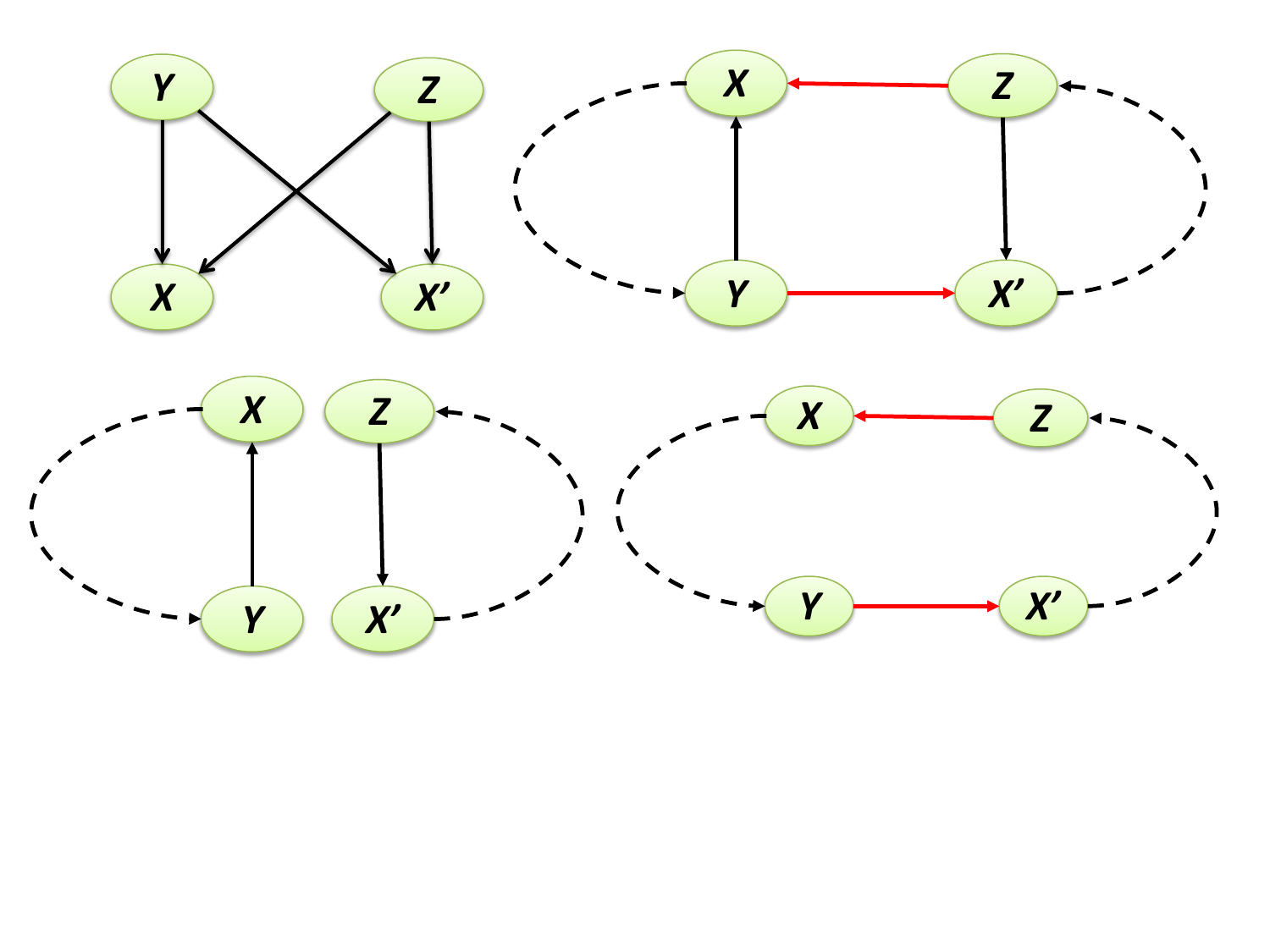}
\end{picture}
\vspace{1.3cm}
\caption{Merging two cycles using companion vertices}
\label{fig:cross-join}
\end{figure}

The merging of all the cycles is done as follows.
The $(n,s)$-necklaces are ordered by their weight from the one with the largest weight to the one with
the smallest weight. We start with the necklace which contains only \emph{ones} and continue to the necklace
with a unique \emph{zero} and start with step 2. At each step, we have a main cycle composed of the $(n,s)$-necklaces that
were merged so far and also the remaining $(n,s)$-necklaces.
At step $i$, $i \geq 2$, we choose on each $(n,s)$-necklace $\cX$ with weight $n-i$ (this is the number of \emph{ones} in each
word of the necklace) a vertex whose representation ends with a \emph{zero}, e.g., $(X0)$. Its companion $(X1)$ on a necklace $\cY$
has weight $n-i+1$ and it is readily verified that $\cY$ does not have a run of more than $s$ \emph{zeros}, i.e., $\cY$ in
also an $(n,s)$-necklace.
Hence, $(X1)$ is on the main cycle, so we can merge the necklace $\cX$ into the main cycle and continue.
This procedure ends when all the $(n,s)$-necklaces are merged into the main cycle.

To apply this procedure, we have to choose on each $(n,s)$-necklace, whose weight is at most $n-2$, a vertex whose
binary representation ends with a \emph{zero}. For example, this vertex can be the one which is the maximum one
as a binary number among all the vertices whose binary representation ends with a \emph{zero}. The companion of this
vertex has a weight greater by one and hence its cycle was merged before and therefore each necklace can be merged in its turn
to the main cycle. This procedure will result in a cycle of $G_n$ which contains all the vertices of the $(n,s)$-necklaces.
Algorithm~\ref{algo:merge} presents the formal steps of this method.

\begin{algorithm}[ht!]
\DontPrintSemicolon

  \KwInput{ Natural numbers $s,n \in \NN$, $s < n-1$}
  \KwOutput{maximum length cyclic $(n,s)$-sequence}
  $i \gets 0$  and $B_i \coloneqq b_i b_{i+1} ~\cdots ~b_{i+n-1} \gets 1^n$\\

   \While{$B_i \neq 0 1~\cdots~1$}
   {
    $Y \gets  b_{i+1} ~ \cdots ~b_{i+n-1} 0$ \\
    \If {the necklace of $Y$ contains a run of more than $s$ zeros}
        {
         $b_{i+n} \gets b_i$ // implies that $b_{i+n} \gets 1$

        }\Else{
         $Z\gets$ cyclic shift of $Y$ with the largest value \\
         \If{$Y = Z$}{
         $b_{i+n} \gets \bar{b}_i$
         } \Else{
         $b_{i+n} \gets b_i$
         }
        }
        $i \gets i+1$
   }

\Return $b_0b_1~\cdots~ b_{i}$ (each $b_i$ is returned when computed)

\caption{merge the $(n,s)$-necklaces using the vertex of largest value in a necklace}
\label{algo:merge}
\end{algorithm}
If we want to form more than one cycle we have to choose on some necklaces two vertices that end in a \emph{zero}.
For $n$ which is not sufficiently large (for example a constant which is not too small)
we will decide on these two necklaces and the two vertices using any chosen criteria. In this case,
$Z$ in the algorithm will be the chosen vertex if $Y$ is on one of these chosen necklaces. Using any simple choice
on a constant number of necklaces, the complexity of this algorithm
is at most $O(n)$ per each bit if we use the algorithm of Booth~\cite{Boo80} to find the cyclic shift with the largest value, which is the
one used for merging in most cases
(in~\cite{Boo80} the smallest value is considered, but the smallest value in a necklace $\cY$ is the complement of the
largest value in $\bar{\cY}$).

If $n$ is large enough and we want to form a larger set of sequences, then
we will choose the necklaces and these vertices that end with a \emph{zero}
based on a storage which is designed in advance and can have different choices as follows.
Choose an integer $k$ which depends on the number of maximum length $(n,s)$-sequences that we want to generate.
On most $(n,s)$-necklaces the word ending in a \emph{zero} with maximum value will be used to merge
the necklace into the main cycle. On $k$ necklaces it will be another word. Each such a word is defined by $\theta (n,s,k)$
free parameters and in total $K={k \cdot \theta (n,s,k)}$ free parameters.
This will imply that $2^K$ maximum length $(n,s)$-sequences will be generated by the algorithm.
The value of $\theta (n,s,k)$ will be increased as $n$ increases, or $s$ increases, or $k$~increases.

For example, let $m=\lceil \log k \rceil$ and consider the ordered set $V=\{ V(i) \}_{i=0}^{k-1}$ of $k$ distinct $(n,s)$-words in $k$
distinct $(n,s)$-necklaces, constructed as follows:
\begin{enumerate}
\item[{\bf (c1)}] The first $(s+1) \lceil \frac{m}{s} \rceil  +1$ bits of $V(i)$ contain \emph{ones}
at positions $(s+1)j$, $0 \leq j \leq \lceil \frac{m}{s} \rceil$ (this guarantee that in these positions there will
not be a run of more than $s$ consecutive \emph{zeros}). In the other
$s \lceil \frac{m}{s} \rceil$ positions (between the positions of the \emph{ones}) $V(i)$ forms the binary representation
of $i$. These $(s+1) \lceil \frac{m}{s} \rceil  +1$ bits are followed by a single \emph{zero}.

\item[{\bf (c2)}] The last $(s+1) \lceil \frac{m}{s} \rceil +4$ bits of $V(i)$ start and end with single \emph{zero}
and between these two \emph{zeros} there are $(s+1) \lceil \frac{m}{s} \rceil +2$ \emph{ones}.
We want that this run of $(s+1) \lceil \frac{m}{s} \rceil +2$ \emph{ones} will be the unique longest run of \emph{ones}
in the word.

\item[{\bf (c3)}] For each $j \geq 1$, in position $\left( (s+1) \lceil \frac{m}{s} \rceil +2 \right)j-1$,
where $\left( (s+1) \lceil \frac{m}{s} \rceil +2 \right)j \leq n-(s+1) \lceil \frac{m}{s} \rceil -3$, there is a \emph{zero}.
This guarantee that the requirement for the longest run of \emph{ones} in {\bf (c2)} is satisfied.

\item[{\bf (c4)}] Between the first $(s+1) \lceil \frac{m}{s} \rceil  +1$ bits and the last $(s+1) \lceil \frac{m}{s} \rceil +4$ bits
each $s+1$ bits we have a \emph{one} (at the end of these $s+1$ positions).
If there is a collision between these \emph{ones} and the \emph{zeros} of {\bf (c3)}, then at the place of the collision we will have
the \emph{one} after $s$ bits and not after $s+1$ bits and after this \emph{one} the \emph{zero} (to avoid such a collision).
This guarantees that there will not be a run of more than $s$ \emph{zeros} in this section of the word.

\item[{\bf (c5)}] The rest of the positions which were not specified are free positions in which there are
arbitrary assignments of \emph{zeros} and \emph{ones} (free parameters).
\end{enumerate}

We first note that each word of $V$ is an $(n,s)$-word from an $(n,s)$-necklace since the word starts with a \emph{one} (by {\bf (c1)})
and ends with an isolated \emph{zero} (by {\bf (c2)}) and in each $s+1$ consecutive positions there is at least one \emph{one}
by {\bf (c1)}, {\bf (c2)}, {\bf (c3)}, and {\bf (c4)}.
Note further that each $V(i)$ has a unique run of exactly $(s+1)\left( \lceil \frac{m}{s} \rceil \right) +2$ \emph{ones}
which ends just one position before the last position. All the other runs of \emph{ones} in each $V(i)$ are shorter.
These properties are guaranteed by {\bf (c2)} and {\bf (c3)}.
This implies that two such words from $V$ cannot be a shift of the other and also all the words are from necklaces of full-order.
In the first $(s+1)\left( \lceil \frac{m}{s} \rceil \right) +1$ bits
there is the binary representation of $i$ (by {\bf (c1)}) with separations of \emph{ones} to satisfy the constraint of no more than
$s$ consecutive \emph{zeros} (again by {\bf (c1)}). This will enable to find if the word that we consider is from the set $V$.
There is a flexibility in the number of free parameters that can be increased as $s$ get larger and/or if we choose a larger $k$.

The algorithm will have similar steps to those of Algorithm~\ref{algo:merge}.
It will first try to see if $Y$ is contained in a necklace with another member of $V$.
In this case the necklace is joined to the main cycle only if $Y$ is the associated word of $V$.
If no word of $V$ is contained in the same necklace as $Y$, then it examines whether $Y$ is a word ending in a \emph{zero} with the maximum value
in its necklace. Note, that in $V$ since the largest run of \emph{ones} is at the end of the word followed by a unique \emph{zero},
it follows that no $V(i)$ can be the word which ends in a \emph{zero} and of maximum value in its necklace
(this longest run of \emph{ones} starts the word of the maximum value in the necklace).
Algorithm~\ref{algo:mergeWithV} presents the formal steps of the described algorithm.

\begin{algorithm}[ht!]
\DontPrintSemicolon

  \KwInput{ Natural numbers $s,n,k \in \NN$, $s < n-1$ and an ordered set $V$}
  \KwOutput{maximum length cyclic $(n,s)$-sequence}
  $i \gets 0$  and $B_i \coloneqq b_i b_{i+1}~\cdots~b_{i+n-1} \gets 1^n$\\

   \While{$B_i \neq 0 1~\cdots~1$}
   {
    $Y \gets b_{i+1}~\cdots~b_{i+n-1}0$ \\
    \If {the necklace of $Y$ contains a run of more than $s$ zeros}
        {
         $b_{i+n} \gets b_i$ // implies that $b_{i+n} \gets 1$

        }\Else{
         \If {$Y$ and some word of $V$ are on the same necklace}
          {\If {Y is a vertex in $V$ } {$b_{i+n} \gets \bar{b}_i$}
          \Else {$b_{i+n} \gets b_i$} }
        \Else{
         $Z\gets$ cyclic shift of $Y$ with the largest value \\
         \If{$Y = Z$}{
         $b_{i+n} \gets \bar{b}_i$
         } \Else{
         $b_{i+n} \gets b_i$
         }
        } }
        $i \gets i+1$
   }

\Return $b_0b_1~\cdots ~b_{i}$ (each $b_i$ is returned when computed)

\caption{merge the $(n,s)$-necklaces with a set $V$}
\label{algo:mergeWithV}
\end{algorithm}

\begin{theorem}
\label{thm:alg_merge_ns}
$~$
\begin{enumerate}
\item[{\bf (a)}] For a given choice of $k$, there are $2^K$, where $K=k \cdot \theta (n,s,k)$, distinct choices for the set
$V$ of stored vertices to merge
the $(n,s)$-necklaces. Thus, Algorithm~\ref{algo:mergeWithV} can be used to produce $2^K$ distinct maximum length $(n,s)$-sequences.

\item[{\bf (b)}] The working space that the procedure requires to produce the next bit of a maximum length $(n,s)$-sequence
is $O(\textup{maximum}\{n,K\})$.

\item[{\bf (c)}] The complexity of the algorithm to find the next bit is $O(n)$.
\end{enumerate}
\end{theorem}
\begin{proof}
The proof of Theorem~\ref{thm:alg_merge_ns} is similar to the one proved in~\cite{EtLe84}.
Each different choice of the ordered set $V$ implies different choices for the vertices via which the merging is performed.
Different choices imply different $(n,s)$-sequences. The working space consists of the $K$ free parameters
of the set $V$ and a constant number $n$ bits
to store the current $n$ bits and $n$ bits to store the necklace and the current shift for comparison.
Hence, the total  working space is $O(\textup{maximum}\{n,K\})$.

As for the time complexity, there are three steps (lines) in the algorithm which are not trivial assignment or comparison.
To compute if there is a run with more than $s$ \emph{zeros} in $Y$ we scan $Y$ and sum the number of consecutive \emph{zeros}
along the necklace. This is done in $O(n)$ time per the $n$ bits of $Y$.
To find if $Y$ and some word of $V$ are on the same necklace, first we compute the length of the
largest run of \emph{ones} in $Y$ in $O(n)$ time per $n$ bits. If this longest run is not of length $(s+1) \lceil \frac{m}{s} \rceil +2$,
then by {\bf (c2)} $Y$ and $V$ are not on the same necklace. Similarly, it there are two such runs of length
$(s+1) \lceil \frac{m}{s} \rceil +2$ in $Y$, then by {\bf (c2)} $Y$ and $V$ are not on the same necklace.
If there is exactly one run of length $(s+1) \lceil \frac{m}{s} \rceil +2$ in $Y$, then we shift $Y$ in a way that
this run with a \emph{zero} before and a \emph{zero} after it will be  at the end of the word. Again, this is done
in at most $O(n)$ time per the $n$ bits.
Let $U$ be the obtained word. Now, if in the first $(s+1) \lceil \frac{m}{s} \rceil  +1$ bits
of $U$ there are \emph{ones} as required of {\bf (c1)}, then the other values in these $(s+1) \lceil \frac{m}{s} \rceil  +1$ bits
indicated the exact entry of $V$ which might be equal $U$ (the shift of $Y$).
Again, for this no more than $O(n)$ time per $n$ bits is required. We compare this entry of $V$ with $U$ to
determine if $Y$ and this word of $V$ are on the same necklace. Again, this is done in $O(n)$ time per $n$ bits.
Finally, to find the cyclic shift of $Y$ with the largest value in the necklace has complexity $O(n)$ due to~\cite{Boo80}.
Thus, the whole process will take no more than $O(n)$ time complexity per a computed bit.
\end{proof}

How large $n$ should be to make this algorithm effective? Recall, that $m=\lceil \log k \rceil$.
The requirement of {\bf (c1)} is $(s+1) \lceil \frac{m}{s} \rceil  +1$ bits, {\bf (c2)}
requires $(s+1) \lceil \frac{m}{s} \rceil  +4$ bits,
{\bf (c3)} requires at most $\frac{n-2(s+1) \lceil \frac{m}{s} \rceil}{(s+1) \lceil \frac{m}{s} \rceil  +2}$ bits,
and {\bf (c4)} requires at most $\frac{n-2(s+1) \lceil \frac{m}{s} \rceil}{s+1} +1$ bits.
This implies that $n$ must be very large, especially if $s$ is small or $k$ is large.
For example, if $s=1$, then only as little as $\frac{n}{2} - 2m -\frac{n-4m}{2m+2} -6$ bits in each word of $V$ are left for the
free parameters. However, when $n$ is sufficiently large the number of generated $(n,s)$-sequences is considerably large.

Other algorithms for generating de Bruijn sequences can be also adapted to merge all the $(n,s)$-necklaces for constructing $(n,s)$-sequences.
The construction of the maximum length $(n,s)$-sequences is based on merging all the $(n,s)$-necklaces.
There are $(n,s)$-sequences that contain $(n,s)$-words that are not contained in $(n,s)$-necklaces.
But, these sequences cannot be of maximum length by Corollary~\ref{cor:exact_neck}.

Finally, we consider a maximum length path in $G_n(s)$. This is done by adding a short path with $(n,s)$-words,
which are not contained in $(n,s)$-necklaces, to the maximum length cycle $\cC$ in $G_n(s)$.
Consider an $(n,s)$-word $(10^s X)$, for some $s$-ones string $X$, whose length is $n-1-s$, on an $(n,s)$-necklace
with the edge $(10^s  X) \rightarrow (0^s X 1)$ of the associated necklace on the cycle $\cC$. We start the maximum length path
with the vertex $(10^s X)$ and continue with the edge $10^s X1$ to the vertex $(0^s X1)$ as on the cycle $\cC$
until we reach the vertex $(10^s X)$ again. So far the path is the same as the
cycle $\cC$. Now, we continue with the $s$ edges,
$$
(10^s X) \rightarrow (0^s X0) \rightarrow ~ \cdots ~ \rightarrow (0X 0^s) ~,
$$
to obtain a path of maximum length in $G_n(s)$, i.e., $s$ edges and $s$ vertices were added to
the maximum length cycle $\cC$ and the outcome is a path in $G_n(s)$ which attain the bound of Theorem~\ref{thm:upper_bound_path}.
All the added vertices are on a necklace $\cX$ which is not an $(n,s)$-necklace, but the added vertices are
represented by the $(n,s)$-words which are contained in $\cX$. Note, that the added sub-path is to
the end of the cycle $\cC$, while in the upper bound it was considered in its beginning, but this is equivalent.

\begin{corollary}
For any $1 \leq s < n-1$ we have that a maximum length path in $G_n(s)$ has length $\ell_{n,s} +s$.
\end{corollary}

\section{Enumeration of the Number of Sequences}
\label{sec:num_seq}

In this section, we will consider enumeration associated with the number of $(n,s)$-words,
$(n,s)$-necklaces, and $(n,s)$-sequences related to the bounds and the
construction that were presented in the previous sections. For this purpose, we define the two well-known important concepts,
the rate and the redundancy of the set of sequences discussed in this paper.
The {\bf \emph{rate}} of an $(n,s)$-sequence $\cS$ with maximum length $N$ will be defined as
$$
\cR(n,\cS)=\limsup_{n \rightarrow \infty} \frac{\log N}{n} ~.
$$
The {\bf \emph{redundancy}} of an $(n,s)$-sequence $\cS$ of length $N$ will be defined as
$$
\text{red}(\cS) = n - \log N ~.
$$
The rate and redundancy are measures for the evaluating the amount of used bit information for the sequence (ratio for the rate
and difference for the redundancy.

The scheme in~\cite{ZQLR21} which was discussed in Section~\ref{sec:intro} uses a de Bruijn sequence of length~$2^n$, where
each \emph{one} is simulated by `11' and each \emph{zero} by `10'. This scheme
considers a sequence of length $2^{n+1}$ with windows of length $2n$ for each required $n$-tuple. Hence, it has a high
redundancy of $2n - \log 2^{n+1} =n-1$ and its rate is 0.5~. Our $(n,s)$-sequences reduce this redundancy
and increase this rate quite dramatically.

What is the exact length of a maximum length $(n,s)$-sequence? The computation of this size with a closed formula can be
done similarly to other computations associated with constrained codes~\cite{Imm90,MRS01}, and asymptotic
computations are done in the same way. Some of the computations can
be done more accurately. For example, we will present the computations for $s=1$.
Let $g_n$~be the number of $(n,1)$-words. The value of $g_n$ is the well-known Fibonacci number which appears
extensively in the theory of constraints codes.

\begin{lemma}
\label{lem:num_n_1}
The number of $(n,1)$-words is
$$
g_n = h_{n,1} = \frac{\left( \frac{1 + \sqrt{5}}{2} \right)^{n+2} - \left( \frac{1 - \sqrt{5}}{2} \right)^{n+2} }{\sqrt{5}}~.
$$
\end{lemma}
\begin{proof}
It is well-known and easily computed that $g_n = g_{n-1} + g_{n-2}$, where $g_1 = 2$ and $g_2 =3$.
The solution for this recurrence is $g_n = \frac{\varphi^{n+2} - \psi^{n+2}}{\sqrt{5}}$,
where $\varphi = \frac{1 + \sqrt{5}}{2}$ and $\psi = \frac{1 - \sqrt{5}}{2}$.
\end{proof}

\begin{lemma}
\label{lem:n1words_in_n1neck}
The number of $(n,1)$-words in the $(n,1)$-necklaces is $\ell_{n,1}=g_n - g_{n-4}$.
\end{lemma}
\begin{proof}
A necklace that is not an $(n,1)$-necklace contains $(n,1)$-words if it starts with 01 and ends with 10.
In between we have an $(n-4,1)$-word and hence by definition there are $g_{n-4}$ such words.
Thus, the number of $(n,1)$-words in the $(n,1)$-necklaces is $\ell_{n,1}=g_n - g_{n-4}$.
\end{proof}

Now we can combine Theorem~\ref{thm:upper_bound_length} and the constructions in Section~\ref{sec:construct_max}
with Lemma~\ref{lem:n1words_in_n1neck} to obtain the following conclusion.

\begin{corollary}
The maximum length of an $(n,1)$-sequence is $g_n - g_{n-4}$.
\end{corollary}
\begin{corollary}
The maximum length of an acyclic $(n,1)$-sequence is $g_n - g_{n-4}+n$, i.e., a trial of length $g_n - g_{n-4}+1$ in $G_{n-1}(s)$.
\end{corollary}

We can now have as a consequence that even a maximum length $(n,1)$-sequence improves the rate of the naive scheme which
was used before with de Bruijn sequences for a quantum key distribution scheme.
\begin{corollary}
The rate of maximum length $(n,1)$ sequences is 0.6942 and their redundancy is $0.3058n$.
\end{corollary}

Lemma~\ref{lem:num_n_1} can be simply generalized to $(n,s)$-words as follows.
\begin{lemma}
\label{lem:num_n_s}
The number $h_{n,s}$ of $(n,s)$-words satisfy the following recursive formula:
$$
h_{n,s} = \sum_{i=1}^{s+1} h_{n-i,s}~,
$$
where $h_{i,s} = 2^i$ for $1 \leq i \leq s$ and $h_{s+1,s}= 2^{s+1}-1$.
\end{lemma}
\begin{proof}
Each $(n,s)$-word can start in $i$ \emph{zeros}, $0 \leq i \leq s$, followed by a \emph{one} and after that there is an $(n-i-1,s)$-word.
This implies the claim in the lemma.
\end{proof}

It was proved in~\cite{CaCu90,LLG15,Nyb12} that $h_{n,s}$ can be expressed as
$$
h_{n,s} = \left\lfloor   \frac{\lambda^{n+1} (\lambda -1)}{(s+2)\lambda - 2(s+1)} + 0.5 \right\rfloor ,
$$
where $\lambda$ is the unique positive real root of the equation
\begin{equation}
\label{eq:char_eq}
x^{s+1} - \sum_{i=0}^s x^i = 0 ~.
\end{equation}
Some values of the maximum length $(n,s)$-sequences and the value of $h_{n,s}$ are presented in Table~\ref{tab:h_n_s}
and the rates of the sequences are presented in Table~\ref{tab:R_n_s},
where $\cR(n,s) = \limsup\limits_{n \rightarrow \infty}\frac{\log N}{n} = \log \lambda$.

\begin{table}[htbp]
    \centering
\setlength\tabcolsep{2pt}
\begin{footnotesize}
\begin{tabular}{|c|c|c|c|c|c|c|c|c|c|c|c|c|c|c|c|c|c|c|c|}
\hline
\diagbox{$s$}{$n$}  & $1$  & $2$  & $3$  & $4 $  & $5 $  & $6 $  & $7  $  & $8  $  & $9  $  & $10  $  & $11  $  \\ \hline
1    & $2$ - $2$ & $3$ - $3$ & $4$ - $5$ & $7 $ - $8 $ & $11$ - $13$ & $18$ - $21$ & $29 $ - $34 $ & $47 $ - $55 $ & $76 $ - $89 $ & $123 $ - $144 $ & $199 $ - $233 $   \\ \hline
2    &  & $4$ - $4$ & $7$ - $7$ & $11$ - $13$ & $21$ - $24$ & $39$ - $44$ & $71 $ - $81 $ & $131$ - $149$ & $241$ - $274$ & $443 $ - $504 $ & $815 $ - $927 $   \\ \hline
3    &  & & $8$ - $8$ & $15$ - $15$ & $26$ - $29$ & $51$ - $56$ & $99 $ - $108$ & $191$ - $208$ & $367$ - $401$ & $708 $ - $773 $ & $1365$ - $1490$  \\ \hline
4    &  &  &  & $16$ - $16$ & $31$ - $31$ & $57$ - $61$ & $113$ - $120$ & $223$ - $236$ & $439$ - $464$ & $863 $ - $912 $ & $1695$ - $1793$   \\ \hline
5    &  &  &  & & $32$ - $32$ & $63$ - $63$ & $120$ - $125$ & $239$ - $248$ & $475$ - $492$ & $943 $ - $976 $ & $1871$ - $1936$   \\ \hline
6    &  &  &  &  &  & $64$ - $64$ & $127$ - $127$ & $247$ - $253$ & $493$ - $504$ & $983 $ - $1004$ & $1959$ - $2000$   \\ \hline
7    &  &  &  &  &  &  & $128$ - $128$ & $255$ - $255$ & $502$ - $509$ & $1003$ - $1016$ & $2003$ - $2028$   \\ \hline
\end{tabular}
\end{footnotesize}
\caption{The maximum length of a cyclic $(n,s)$-sequence and the number of $(n,s)$-words for $1\leq n\leq 11$ and $1\leq s\leq \text{maximum}(7,n)$.}
\label{tab:h_n_s}
\end{table}

\begin{table}[htbp]
    \centering
\begin{tabular}{|c|c|c|c|c|c|c|}
\hline
$s$             & $1$      & $2$      & $3$      & $4$      & $5$      & $6$        \\ \hline
$\lambda$       & $1.6180$ & $1.8393$ & $1.9276$ & $1.9659$ & $1.9836$ & $1.9920$  \\ \hline
$\log \lambda$ & $0.6942$  & $0.8791$ & $0.9468$   & $0.9752$ & $0.9881$ & $0.9942$  \\ \hline
\end{tabular}

\vspace{0.2cm}

\begin{tabular}{|c|c|c|c|c|c|c|}
\hline
$s$               &  $7$      & $8$      & $9$      & $10$     & $11$     & $12$     \\ \hline
$\lambda$       & $1.9960$ & $1.9980$ & $1.9990$ & $1.9995$ & $1.9998$ & $1.9999$ \\ \hline
$\log \lambda$ &  $0.9971$ & $0.9986$ & $0.9993$ & $0.9996$ & $0.9998$ & $0.9999$ \\ \hline
\end{tabular}
\caption{The rate, $\cR(n,s)$, of maximum length $(n,s)$-sequences for $1 \leq s \leq 12$.}
\label{tab:R_n_s}
\end{table}

Lemma~\ref{lem:n1words_in_n1neck} can be generalized as follows for $s >1$.


\begin{lemma}
The number of $(n,s)$-words in the $(n,s)$-necklaces is
$$
\ell_{n,s}=h_{n-1,s} + \sum_{i=1}^s (i \cdot h_{n-i-2,s}) ~,
$$
where $h_{0,s}=1$.
\end{lemma}
\begin{proof}
If an $(n,s)$-word starts with a \emph{one}, then after this \emph{one} we can have any one of the $h_{n-1,s}$ $(n-1,s)$-words.
Otherwise, an $(n,s)$-word is an $(n,s)$-necklaces has a prefix $0^{i_1}1$ and a suffix $1 0^{i_2}$, where $i_1 \geq 1$ and
$i_1+ i_2 = i \leq s$, i.e., $1 \leq i \leq s$,
and between them there could be any one of the $h_{n-i-2,s}$ $(n-i-2,s)$-words which implies the claim of the lemma.
\end{proof}

\begin{corollary}
The maximum length of a cyclic $(n,s)$-sequence is $\ell_{n,s}=h_{n-1,s} + \sum_{i=1}^s (i \cdot h_{n-i-2,s})$.
The maximum length of an acyclic $(n,s)$-sequence is $\ell_{n,s}=h_{n-1,s} + \sum_{i=1}^s (i \cdot h_{n-i-2,s})+s+n-1$, where $s \leq n-2$.
\end{corollary}
\begin{corollary}
The rate of maximum length $(n,s)$-sequences is $\log \lambda$ and their redundancy is $n - \lambda n$, where $\lambda$
is the root of Eq.~\textup{(\ref{eq:char_eq})}.
\end{corollary}

We might be interested in the number of
$(n,s)$-necklaces of weight $k$. This can be associated with an
algorithm to merge $(n,s)$-necklaces, with restricted weight (another possible
constraint), especially when $s=1$. When $s=1$ this number is not difficult to compute.

\begin{lemma}
\label{lem:n_1words}
The number of $(n,1)$-words with weight $k$ is $\binom{k+1}{n-k}$ if $k+1 \geq n-k$.
\end{lemma}
\begin{proof}
An $(n,1)$-word with weight $k$ has $k$ \emph{ones} and the $n-k$ \emph{zeros} are isolated between
the \emph{ones} (including at the beginning or the end) and hence there are $\binom{k+1}{n-k}$ such words.
\end{proof}

\begin{lemma}
\label{lem:weight_k_neck}
The number of words in the $(n,1)$-necklaces with words of weight $k$ is
$$
\binom{k}{n-k} + \binom{k-1}{n-k-1}
$$
if $k+1 \geq n-k$.
\end{lemma}
\begin{proof}
A word in an $(n,1)$-necklace is an $(n,1)$-word that does not start and end in a \emph{zero}.
By Lemma~\ref{lem:n_1words} the total number of $(n,1)$-words with weight $k$ is $\binom{k+1}{n-k}$.
An $(n,1)$-word which starts with a \emph{zero} and ends with a \emph{zero}, starts with 01 and ends with 10.
In between we have an $(n-4,1)$-word with weight $k-2$. By Lemma~\ref{lem:n_1words}, there are $\binom{k-1}{n-k-2}$ such words and hence
the number of words in the $(n,1)$-necklaces with words of weight $k$ is
$$
\binom{k+1}{n-k} - \binom{k-1}{n-k-2}= \binom{k}{n-k} + \binom{k-1}{n-k-1} ~.
$$
\end{proof}

It should be noted that on one hand the set $V$ defined in Section~\ref{sec:construct_max} is less effective
when $s=1$, but based on Lemmas~\ref{lem:n_1words} and~\ref{lem:weight_k_neck} an efficient algorithm to construct
a large set of maximum length $(n,1)$-sequences can be designed. This algorithm will be based on
Algorithm~\ref{algo:merge} to merge $(n,1)$-necklaces and an efficient algorithm to enumerate the associated binomial coefficient, e.g.,
the enumerative encoding of Cover~\cite{Cov73}.

\section{Sequences with Efficient Positioning Decoding}
\label{sec:decode}

In this section, we use some known algorithms and present a very simple and efficient algorithm for
generating one $(n,s)$-sequence (acyclic and cyclic) of maximum length.
The advantage of this algorithm is that the position of any given ${n\text{-tuple}}$ can be decoded efficiently.
The algorithm is based on an idea of concatenating necklaces that was presented by Fredricksen and Maiorana~\cite{FrMa78}
which improved on a previous
idea of Fredricksen and Kessler~\cite{FrKe77} that generated a de Bruijn sequence based on a partition of $n$ into smaller positive integers.
The algorithm was improved later by Fredricksen and Kessler~\cite{FrKe86} and
it has a few variants. It is natural to call this algorithm, the FKM algorithm for Fredricksen, Kessler, and Maiorana.
The variant of the algorithm that we consider uses representatives of the necklaces which are called Lyndon words.
For a given necklace of order~$n$, its {\bf \emph{Lyndon word}} of order $n$ is the word of the least value in the necklace,
whereas for a full-order necklace of length $n$, a word of length $n$ is taken, while for a degenerated necklace
of length $d < n$ which divides $n$, a word of length $d$ is taken. For example, when $n=6$ and the necklace contains
the words $(010101)$ and $(101010)$ the Lyndon word is $01$, i.e., one period of the sequence.
The Lyndon words are now ordered
lexicographically from the smallest to the largest one and concatenated together in this order.
The outcome is a de Bruijn sequence of length $2^n$ and it is called the lexicographically least
de Bruijn sequence. In the original papers~\cite{FrKe86,FrMa78} the word of maximum value in base 2 was taken from a necklace and
the necklaces were ordered from the maximum value to the minimum value in this representation.
Hence, the generation of the associated de Bruijn sequence
is going down to the ordering of the necklaces lexicographically based on their Lyndon words.
But, to generate an $(n,s)$-sequence of maximum length we have to use only the $(n,s)$-necklaces and therefore
we have to make a small modification to the original sequence and also to the original algorithm.
As mentioned, the original sequences generated in~\cite{FrKe86,FrMa78} considered the necklaces ordered by the words with the largest binary
value from the largest one to the smallest one, but later algorithms with the same
technique used a different order for the Lyndon words, e.g.~\cite{SWW14,SWW16}.
The algorithm was analyzed by Ruskey, Savage, and Wang~\cite{RSW92}.
Using the ranking and the decoding algorithm for these sequences as
suggested by Kociumaka, Radoszewski, and Rytter~\cite{KRR16} to decode the lexicographic least de Bruijn sequence,
we can decode the generated $(n,s)$-sequence which will be constructed in this section.
A general framework for concatenation of necklaces and in particular necklaces which avoid certain patterns
was given in~\cite{GaSa18,SWW16}. Such algorithms for generating necklaces and strings
with forbidden substrings were also given in~\cite{RuSa00}.
The algorithm which are presented in these papers and especially those given in~\cite{GaSa18,SWW16} can be applied directly
for $(n,s)$-sequences and can be implemented in practice to form the required sequences.
A~more recent algorithm which combine merging of necklaces and
concatenation of necklaces and should be mentioned was presented in~\cite{SSTW23}.
These algorithms can be implemented in $O(n)$ time to construct the next bit and in average with $O(1)$ time per bit~\cite{SWW16X} using
$O(n)$ space. Practical implementation of the decoding ideas, to find the position of a given
$n$-tuple can be found in~\cite{SaWi17}. Finally, we would like to mention
that greedy algorithms with different successor rules~\cite{GSWW18} can also be used to generate one $(n,s)$-sequence.

The Lyndon words have some simple properties which were used to merge all the necklaces into a de Bruijn sequence~\cite{FrKe77,FrKe86,FrMa78}.
\begin{lemma}
\label{lem:propLyndon}
The Lyndon word of a nonzero necklace starts with the longest run of \emph{zeros} and ends with a \emph{one}.
\end{lemma}
\begin{lemma}
\label{lem:lead0Lyndon}
If the Lyndon word $X_1$ of a necklace has a larger run of consecutive \emph{zeros} than in the Lyndon word $X_2$ of another necklace,
then $X_1$ has a smaller value than $X_2$.
\end{lemma}

Lemma~\ref{lem:propLyndon} is not a precise characterization of a Lyndon word since for example there might be a few runs with
the longest run of \emph{zeros}.

The least lexicographic de Bruijn sequence is generated based on the following celebrated lemma~\cite{FrMa78}.

\begin{lemma}
\label{lem:listLyndon_DB}
Let $X_1, X_2, X_3, ~ \ldots, ~ X_\ell$ be the ordering of the Lyndon words of order $n$ from the smallest one lexicographically to
the largest one. The concatenation of these words in this order is a de Bruijn sequence of order $n$, i.e., an Eulerian circuit
in $G_{n-1}$ which is also a Hamiltonian cycle in $G_n$.
\end{lemma}

The next step is to show that lemma~\ref{lem:listLyndon_DB} is also true if we restrict ourselves to the Lyndon words of the $(n,s)$-necklaces.
This is the simple idea which led to efficient construction of $(n,s)$-sequences from all the $(n,s)$-necklaces~\cite{GaSa18,RSW92,SWW16}.
The idea is summarized in the following lemmas.

\begin{lemma}
\label{lem:parts_lyndon}
Let $X_1, X_2, X_3, ~ \ldots, ~ X_\ell$ be the ordering of the Lyndon words of order~$n$ from the smallest one lexicographically to
the largest one. There exists an index $k$ such that $X_1, X_2, ~ \ldots, ~ X_k$ are Lyndon words which are not contained
in $(n,s)$-necklaces, while $X_{k+1}, X_{k+2}, ~ \ldots, ~ X_\ell$ are Lyndon words which are contained in $(n,s)$-necklaces.
\end{lemma}

\begin{corollary}
\label{cor:last_Lyndon}
The last Lyndon word, which is not an $(n,s)$-word, in the order from the smallest one to the largest one, is $0^{s+1} 1^{n-s-1}$.
\end{corollary}

\begin{corollary}
\label{cor:all_after}
In the lexicographic ordering of the Lyndon words of order $n$,
all Lyndon words after the Lyndon word $0^{s+1} 1^{n-s-1}$ are contained in $(n,s)$-necklaces.
\end{corollary}

The algorithm, which generates all the Lyndon words that follow the last Lyndon word $0^{s+1} 1^{n-s-1}$ which is
not contained in an $(n,s)$-necklace, is presented in Algorithm~\ref{algo:lexicographic}. It follows very similar steps to
the ones in Ruskey, Savage, and Wang~\cite{RSW92}.




%
%
%
%

\begin{algorithm}[ht!]
\DontPrintSemicolon

  \KwInput{$n$, $s$, $X_0 \coloneqq 0^{s+1}1^{n-s-1}$}
  \KwOutput{maximum length lexicographic cyclic $(n,s)$-sequence}
  Set $Y = y_1y_2~\cdots ~y_n \gets X_0$ and $i \gets 0$ \\
   \While{$Y \neq 1^{n}$}
   {
   	$j \gets \max \{ t \in \{1,2,\dots, n\}: y_t = 0\} $\\
        $Z \gets y_1y_2~\cdots ~ y_{j-1}1$ \\
        $V = v_1v_2v_3\dots \gets ZZZ\cdots $\\
        $Y \gets v_1v_2~\cdots~ v_n$\\
        \If{$j$ divides $n$ }
        {
         $i \gets  i+1$ \\
        $X_i \gets Z$

        }
   }
   \Return $X_1X_2~\cdots ~X_i$ (each $X_i$ is returned when computed)

\caption{Lexicographic generation of Lyndon words in the $(n,s)$-necklaces}
\label{algo:lexicographic}
\end{algorithm}

The output of the algorithm is a maximum length cyclic $(n,s)$-sequence. By Theorem~\ref{thm:upper_bound_length},
Lemma~\ref{lem:parts_lyndon}, and Corollaries~\ref{cor:last_Lyndon} and~\ref{cor:all_after},
the concatenation of $X_1 X_2 ~ \cdots ~ X_i$ is the required maximum length cyclic $(n,s)$-sequence that contains exactly all the words
of the $(n,s)$-necklaces. For the acyclic $(n,s)$-sequence we had to start after the word $0^{s+1} 1^{n-s-1}$ and this
is the word associated with the first $n$ bits of $0^s 1^{n-s-1}0$ since the first $n$ bits of $X_1$ are \emph{zeros} followed
by a \emph{one}. This also implies that for the acyclic $(n,s)$-sequence after $X_i$ we have to add $s$ \emph{zeros}
which add $s$ edges to the sequence (see the discussion that follows Theorem~\ref{thm:upper_bound_length}).
Hence, we have the following result which is the main theorem of this section.
\begin{theorem}
Algorithm~\ref{algo:lexicographic} produces maximum length cyclic $(n,s)$-sequences and a maximum length acyclic
$(n,s)$-sequence is generated if in the last line it will return the sequence $0^s 1^{n-s-1}X_1X_2\dots X_i 0^s$.
\end{theorem}

The correctness of the algorithm is shown in the same way as it is proved in the FKM algorithm.
To find the position of a given word $v$ of length $n$ in the concatenation of the Lyndon words we have to apply
the algorithm proposed by Kociumaka, Radoszewski, and Rytter~\cite{KRR16} which finds the position of $v$ in the concatenation
of all the Lyndon word of order~$n$. After the position of $v$ was found we have to subtract the position
of the last entry in the word $0^{s+1} 1^{n-s-1}$ since our sequence starts after this last entry.
It is worthwhile and more efficient in the long run
to compute this position in advance and save it as it will be used in every application of the algorithm.
The complexity of the decoding algorithm is the same as in~\cite{KRR16} since the ordering of the Lyndon words
is the same with the exception that we do not start from the first one but with the one after $0^{s+1} 1^{n-s-1}$.
Finally, the average complexity of computing the next bit is constant using the techniques and
algorithms as was explained first in~\cite{RuSa00}.

\section{Conclusion and Discussion}
\label{sec:conclude}

Motivated by an application for space-free quantum key distribution a system based on a simple run-length
limited sequences in the de Bruijn graphs is proposed. The maximum length of such sequences is shown to
be associated with the number of constrained necklaces. An efficient algorithm to generate
a large set of such sequences is proposed and some enumerations related to the length of a maximum length
sequences are discussed. Known algorithms to generate one such sequence efficiently are mentioned.
Generalizations for larger alphabet or for sequences in which each window of length $n$ has a constrained weight
can be easily derived from our exposition.

\end{document}